\documentclass[journal,doublecolumn]{IEEEtran}
\usepackage{cite}
\usepackage[english]{babel}
\usepackage{float}
\usepackage{stfloats}
\usepackage{amsthm}
\usepackage{amssymb}
\usepackage{tikz}
\usetikzlibrary{shapes,arrows,positioning,calc}
\usepackage{commath}
\usepackage{mathrsfs}
\usepackage{mdframed}
\usepackage{multirow}
\usepackage{graphicx}
\usepackage{booktabs}
\usepackage{tabularx}
\usepackage{subcaption}
\usepackage{array}
\usepackage{comment}
\usepackage{relsize,amsfonts}
\usepackage{mathtools}
\usepackage{algorithm}
\usepackage{algorithmic}
\usepackage{bbm}
\usepackage[colorlinks=true, allcolors=blue]{hyperref} 

\newtheorem{Theorem}{Theorem}
\newtheorem{Lemma}{Lemma}
\newtheorem{Corollary}{Corollary}
\newtheorem{Proposition}{Proposition}
\newtheorem{Assumption}{Assumption}
\newtheorem{Remark}{Remark}
\newtheorem{Definition}{Definition}

\newcommand{\E}{\mathbb{E}}
\newcommand{\R}{\mathbb{R}}
\newcommand{\argmin}{\mathop{\mathrm{arg\,min}}}

\newcommand{\simplex}{\Delta}
\newcommand{\clients}{\mathcal{N}}
\newcommand{\fast}{\mathcal{F}}
\newcommand{\slow}{\mathcal{S}}

\newcommand{\Ls}{L_{\mathrm{sm}}}      
\newcommand{\mc}{m_{\mathrm{cv}}}      
\newcommand{\scoef}{\omega}            
\newcommand{\sconst}{\nu}              
\newcommand{\ba}{\boldsymbol{a}}
\newcommand{\bb}{\boldsymbol{b}}
\newcommand{\bw}{\boldsymbol{w}}

\title{Tail-Latency-Aware Federated Learning with Pinching Antenna: Latency, Participation, and Placement}

\author{Yushen Lin,~\IEEEmembership{Student Member,~IEEE,}
       and {Zhiguo Ding,~\IEEEmembership{Fellow,~IEEE}}
\vspace{-2em} 
\thanks{Y. Lin is with the School of Electrical and Electronic Engineering, The University of Manchester, M13 9PL, U.K. (e-mail: yushen.lin@manchester.ac.uk).
Zhiguo Ding is with the School of Electrical \& Electronic Engineering, Nanyang Technological University, 639798, Singapore;  (e-mail: zhiguo.ding@ntu.edu.sg).}
}

\begin{document}

\maketitle

\begin{abstract}
Straggler synchronization is a dominant wall-clock bottleneck in synchronous wireless federated learning (FL). Under non-IID data, however, aggressively sampling only fast clients may significantly slow convergence due to statistical heterogeneity. This paper studies PASS-enabled FL, where a radiating pinching antenna (PA) can be activated at an arbitrary position along a dielectric waveguide to reshape uplink latencies. We consider a joint optimization of PA placement and client participation to minimize the expected time-to-accuracy, coupling the exact expected maximum round latency via order statistics with a heterogeneity-aware convergence factor. We derive first-order optimality conditions that reveal an explicit tail-latency premium in the KKT recursion, quantifying how latency gaps are amplified by maximum-order-statistic synchronization.
Under a latency-class structure, we obtain a within-class square-root sampling law and establish a two-class phase transition where slow-class participation collapses under an explicit heterogeneity-threshold condition as the per-round sample size grows.
For PA placement, we prove a piecewise envelope-derivative characterization and provide an exact breakpoint-and-root candidate-enumeration procedure. Simulation results verify the theoretical findings and show that PASS enables more eligible participation, yielding higher wall-clock accuracy.
\end{abstract}

\section{Introduction}
Federated learning (FL) has attracted research interest due to its ability to address data privacy concerns while enabling collaborative learning over large-scale distributed data, as well as alleviating the heavy computational burden on end clients \cite{FL_original, Zhaohui_FL_survey}.
However, wireless FL is often wall-clock limited by stragglers: in synchronous FL, each round must wait for the slowest selected uplink (the ``straggler''), and the slowest-link tail dominates the overall training time \cite{FL_Wireless_Qin}. 
Meanwhile, non-IID data create statistical heterogeneity, making aggressive selection of ``fast'' devices potentially harmful to convergence \cite{FL_with_non_iid}. 

Pinching-antenna systems (PASS) have recently gained significant research attention as a physical-layer mechanism to control the propagation environment \cite{fang_PASS}. PASS operates by dynamically activating radiating pinching antennas (PAs) along a dielectric waveguide. By activating PAs at different positions, PASS can create line-of-sight (LoS) links and establish short "last-meter" connections \cite{Ding2025FlexibleAntenna, Liu2025PASSTutorial}, which has the potential to mitigate stragglers by shortening the worst links and stabilizing LoS connectivity in challenging high-frequency environments \cite{Lin2025PASSFL, Liu2025PASSOutlook}.

\paragraph{PASS}
At a foundational level, \cite{Ding2025FlexibleAntenna} studies pinching antennas within the broader class of flexible-antenna systems, developing analytical results for the basic single-waveguide/single-PA setting to quantify how position control can mitigate large-scale path loss by creating strong LoS links.
Since PA location is the main control knob in PASS, placement optimization has also been studied to obtain sharper insights.
In \cite{DingPoor2025Placement}, closed-form placement solutions are derived for several multiple-access objectives.
For user-fairness-oriented OMA, the optimal PA activation location must be beneficial to all served users and does not depend on users' perpendicular distances to the waveguide \cite{DingPoor2025Placement}.

Joint resource allocation is challenging because PA locations couple with conventional variables such as transmit power and time/frequency resources \cite{Zeng2025RAPASS}.
The work in \cite{Zhao2025MultiUserPASS} proposes multi-waveguide transmission structures, including waveguide multiplexing, waveguide division, and waveguide switching. It formulates max--min fairness beamforming problems for multi-group multicast, and develops penalty dual decomposition-based algorithms to handle the complex exponential and fractional coupling induced by pinching beamforming \cite{Zhao2025MultiUserPASS}.
In \cite{Li2025PASSISAC}, a multi-waveguide PASS-assisted ISAC architecture is proposed, and sensing accuracy is measured by a CRB metric.
The resulting CRB-minimization problem is solved by alternating between SDR-based digital beamforming and penalty-based PA deployment updates under communication QoS and deployment constraints \cite{Li2025PASSISAC}.

\paragraph{Straggler and dropout mitigation in federated learning}
In synchronous FL, the wall-clock time of each round is dictated by the slowest selected client, which makes training sensitive to stragglers.
A large body of work mitigates this effect at the system and protocol layers, often without changing the underlying wireless channel.
Client selection and scheduling are common approaches: Oort prioritizes clients using a utility-and-speed signal to improve training efficiency \cite{Lai2021Oort}, while TiFL uses a tier-based design to reduce waiting time and to manage device heterogeneity \cite{Chai2020TiFL}.
Another direction relaxes strict synchronization.
FedBuff introduces buffered asynchronous aggregation to reduce idle waiting and to better tolerate stragglers \cite{Nguyen2022FedBuff}.
Recent wireless-edge work also uses straggler-aware grouping and clustering to stabilize training under heterogeneous link qualities \cite{Liu2024FeDSC}.
From a learning-theoretic angle, straggler-resilient FL has been studied by explicitly balancing statistical accuracy with system heterogeneity, which clarifies when faster participation helps or hurts convergence \cite{Reisizadeh2022StragglerResilient}.
These methods are effective in many settings, but they typically treat the channel and latency distribution as fixed.
As a result, they often rely on biased participation (favoring fast clients) or asynchrony, both of which can interact poorly with non-IID data.

\paragraph{PASS-empowered edge learning and PASS-FL}
Unlike the above approaches, PASS can reshape the physical-layer latency landscape itself by controlling where radiation occurs.
This capability has made PASS attractive for edge intelligence, where training time is dominated by extreme straggler rounds rather than average latency.
A broader view on the synergy between pinching antennas and AI argues that AI can control pinching positions, while pinching antennas can support edge-AI workloads such as FL and over-the-air aggregation \cite{Fang2025PinchingAI}.
At the aggregation level, PASS has been studied for AirComp to reduce channel misalignment and improve aggregation accuracy through joint placement and communication design \cite{Lyu2025PASSAirComp}.
For FL specifically, recent work analyzed PASS-enabled FL in terms of tail latency, participation probability, and convergence behavior, highlighting that addressing stragglers at the PHY layer can complement higher-layer scheduling \cite{Lin2025PASSFL}.
Related studies also considered hybrid networks that combine conventional access and pinching antennas, and optimized placement and resource allocation to improve communication efficiency for FL \cite{Wu2025HCPAN}.
However, existing PASS-FL works focus primarily on demonstrating latency or participation improvements through system-level optimization; they do not provide analytical characterizations of the optimal participation--latency tradeoff under non-IID data or derive structural properties of the joint placement-and-sampling problem.

This paper addresses these gaps with the following contributions:
\begin{itemize}
    \item  We propose a cross-layer time-to-accuracy objective for PASS-enabled wireless FL under non-IID data by coupling exact straggler latency with a heterogeneity-aware convergence factor.
    \item For a fixed PA location, we characterize locally optimal client sampling via KKT conditions and reveal a tail-latency premium that quantifies straggler amplification across latency gaps.
    \item We derive a class-wise reduction that yields a within-class square-root law, and establish a two-class phase transition where slow-class participation collapses to the $1/K$ scale under explicit conditions.
    \item We provide an exact piecewise envelope-derivative characterization for PA placement and develop a breakpoint-and-root candidate search that is globally optimal for the outer envelope objective. 
    \item Simulations validate the phase transition, demonstrate Pareto-dominant operating points in the latency--convergence trade-off, and show substantial accuracy gains under deadline-constrained FL.
\end{itemize}

\paragraph{Organization.}
Section~\ref{sec:model} introduces the PASS link/latency model and the order-statistics straggler functional.
Section~\ref{sec:kkt} studies optimal client participation for a fixed PA location, derives the
tail-latency-aware KKT conditions, and establishes the latency-class reduction and
phase transition.
Section~\ref{sec:placement} characterizes the PA placement problem via a piecewise envelope derivative and presents a breakpoint-and-root candidate search algorithm.
Section~\ref{sec:simulation} provides numerical experiments, followed by conclusions in Section~\ref{sec:conclusion}.

\section{System Model} \label{sec:model}
We consider a single-waveguide PASS in which a radiating pinching antenna is activated at position $x$ along the waveguide, where $L>0$ denotes the length of the dielectric waveguide, so that the PA activation position satisfies $x\in[0,L]$.
The Single-antenna client $i$ has the projection of the waveguide $u_i\in\R$ and the effective transverse distance $r_i>0$ to the waveguide, where $i \in \mathcal{N}\triangleq\{1,2,\ldots,N\}$ denotes the set of clients and $N$ is the total number of clients. The PA--client distance is
\begin{equation}
    d_i(x)^2=(x-u_i)^2+r_i^2.
\end{equation}
Following the standard PASS link model \cite{Ding2025FlexibleAntenna,Lin2025PASSFL}, the end-to-end uplink channel gain consists of free-space LoS spherical-wave propagation from the PA to each client. Specifically, the complex baseband channel from the PA at $x$ to client $i$ is modeled as \cite{Lyu2025PASSAirComp}
\begin{subequations}\label{eq:pass_link}
\begin{align}
h_i(x)
&= \sqrt{\eta_f}\,
\frac{e^{-j k_0 d_i(x)}}{d_i(x)}, \label{eq:pass_link_h}\\
\gamma_i(x)
&= \frac{P_i}{\sigma^2}\,|h_i(x)|^2, \label{eq:pass_link_snr}
\end{align}
\end{subequations}
where $P_i$ is the uplink transmit power of client $i$ and $\sigma^2$ is the receiver noise power.
Here, $\eta_f \triangleq \frac{c_0^2}{16\pi^2 f_c^2}$ is the Friis free-space factor,
and $k_0 \triangleq 2\pi f_c/c_0$ is the free-space wavenumber.

Each selected client uploads $S_i$ bits over bandwidth $B$. The per-round latency of client $i$ is
\begin{equation}
    t_i(x)=t_i^{\mathrm{comp}}+\tau_i(x),
    \qquad
    \tau_i(x)=\frac{S_i}{B\log_2\big(1+\gamma_i(x)\big)},
    \label{eq:ti}
\end{equation}
where $t_i^{\mathrm{comp}}$ denotes the local computational time of client $i$.

For a fixed PA position $x$, we relabel clients so that their per-round latencies satisfy
\begin{equation}
    t_1(x)\le t_2(x)\le \cdots\le t_N(x).
    \label{eq:sorted-t}
\end{equation}
All quantities indexed by $i$ refer to this sorted order. Through \eqref{eq:pass_link}--\eqref{eq:ti}, PA placement $x$ determines the uplink SNR $\gamma_i(x)$ via $d_i(x)$, and thus $t_i(x)$. Varying $x$ reshapes $\{t_i(x)\}_{i=1}^N$ and reorders clients by latency, i.e., the identity of the slowest links can change with $x$. Since synchronous FL waits for the maximum latency among the $K$ sampled clients, the distribution of the straggler depends on the relative ordering of $\{t_i(x)\}$ under the sampling probabilities $q$. This motivates the sorted relabeling in \eqref{eq:sorted-t} and the order-statistics treatment below.

\subsection{Straggler Order Statistics}
In each round, the server samples $K$ clients i.i.d.\ with replacement according to a distribution $q=(q_1,\dots,q_N)$, where
\begin{equation}
    q\in\simplex\triangleq\big\{q\in\R^N:\ q_i>0,\ \sum_{i=1}^N q_i=1\big\}.
    \label{eq:simplex}
\end{equation}
Let $\mathcal{K}(q)$ denote the multiset of sampled clients. The per-round wall-clock time is determined by the straggler:
$T_{\mathrm{round}}(q,x)=\max_{i\in\mathcal{K}(q)} t_i(x)$.
Define the cumulative sampling mass up to client $i$ in sorted order \eqref{eq:sorted-t}:
\begin{equation}
    Q_i\triangleq \sum_{j=1}^i q_j,
    \qquad Q_0\triangleq 0.
    \label{eq:Qi}
\end{equation}
Then it is straightforward to obtain that for fixed $x$ and sorted times \eqref{eq:sorted-t},
\begin{equation}
    f(q,x)\triangleq\E\big[T_{\mathrm{round}}(q,x)\big]
    =\sum_{i=1}^N\big(Q_i^K-Q_{i-1}^K\big)\,t_i(x).
    \label{eq:f-order}
\end{equation}

\subsubsection{Tail-Latency Coupling}
Before analyzing the general $N$-client problem, we illustrate the mechanism of tail-latency amplification using a two-class model.
Consider a fast class with latency $t_f$ and a slow class with latency $t_s=t_f+\Delta$.
Let $\delta\in[0,1]$ be the total sampling mass assigned to the slow class. Under i.i.d.\ sampling:
\[
\mathbb{E}[T_{\mathrm{round}}]
= t_f + \Delta\big(1-(1-\delta)^K\big),
\qquad
\frac{\partial}{\partial \delta}\mathbb{E}[T_{\mathrm{round}}]
=\Delta K(1-\delta)^{K-1}.
\]
When $\delta K\ll 1$, we have $(1-\delta)^{K-1}=1+O(\delta K)$ and thus
\(
\frac{\partial}{\partial \delta}\mathbb{E}[T_{\mathrm{round}}]
=\Delta K\,(1+O(\delta K)).
\)

The term $1-(1-\delta)^K$ is the probability that at least one draw lands in the slow class. The factor $(1-\delta)^{K-1}$ is the probability that the other $K-1$ draws remain in the fast class, so the marginal cost
$\frac{\partial}{\partial\delta}\E[T_{\mathrm{round}}]=\Delta K(1-\delta)^{K-1}$
can be read as
\emph{(gap height)} $\times$ \emph{($K$-fold amplification)} $\times$ \emph{(``near-miss'' probability)}.
This also reveals the natural participation scale $\delta=\rho/K$: then $(1-\delta)^K\to e^{-\rho}$ and the slow class ceases to be an almost-sure straggler.

This two-class intuition of $K$-fold amplification of latency gaps extends to the general $N$-client setting and is formalized in Section~\ref{sec:kkt}.

\subsubsection{Coupling with Non-IID Convergence}

Consider the global FL objective
\begin{equation}
    F(\bw)=\sum_{i=1}^N p_i F_i(\bw),
    \qquad \sum_{i=1}^N p_i=1,\ p_i>0,
\end{equation}
where $p_i$ denotes the target aggregation weight, e.g. $p_i = n_i/\sum_j n_j$, and
client $i$ holds a local dataset $\mathcal{D}_i=\{\xi_{i,1},\ldots,\xi_{i,n_i}\}$ of size $n_i$. The local empirical risk is
\begin{equation}
    F_i(\bw)\triangleq \frac{1}{n_i}\sum_{k=1}^{n_i}\ell(\bw;\xi_{i,k}),
    \qquad i\in\{1,\ldots,N\}.
\end{equation}
where $\ell(\bw;\xi)$ denote the per-sample loss function. The goal of FL is to solve
\begin{equation}
    \min_{\bw\in\mathbb{R}^d} F(\bw).
\end{equation}
\noindent For notational convenience, define the global and local optimal values
\(
    F^\star \triangleq \min_{\bw} F(\bw), \
    F_i^\star \triangleq \min_{\bw} F_i(\bw),
\)
and let $\bw^\star$ be any minimizer of $F$.
We denote by $\nabla F_i(\bw;\xi_i)$ the stochastic gradient computed from a random sample
$\xi_i$ (or a mini-batch) drawn from $\mathcal{D}_i$, i.e.,
$\nabla F_i(\bw;\xi_i)=\nabla \ell(\bw;\xi_i)$ and
$\mathbb{E}_{\xi_i}[\nabla F_i(\bw;\xi_i)]=\nabla F_i(\bw)$.
In each round $r$, the server samples $K$ clients i.i.d.\ according to $\boldsymbol{q}=(q_1,\ldots,q_N)\in\Delta$, broadcasts $\bw^{(r)}$ under FedAvg.
Each selected client $i$ performs $E$ local SGD steps:
\begin{align}
\bw_i^{(r,0)} &= \bw^{(r)},\\
\bw_i^{(r,e+1)} &= \bw_i^{(r,e)}-\eta\,\nabla F_i(\bw_i^{(r,e)};\xi_i^{(r,e)}),
\qquad e=0,\ldots,E-1,
\end{align}
and returns $\bw_i^{(r,E)}$.

\begin{Assumption}\label{ass:hetero_penalty}
For each client $i$, the local objective $F_i$ is $\Ls$-smooth and $\mc$-strongly convex:
for all $\ba,\bb$,
\begin{align}
\|\nabla F_i(\ba)-\nabla F_i(\bb)\| &\le \Ls \|\ba-\bb\|,\\
F_i(\ba) &\ge F_i(\bb)+\nabla F_i(\bb)^\top(\ba-\bb)+\frac{\mc}{2}\|\ba-\bb\|^2.
\end{align}
\end{Assumption}

\begin{Assumption} \label{ass:2}
Assume that for each client $i$ there exist constants $G_i$ and $\chi_i$ such that the variance of the stochastic gradient and the expected squared norm of stochastic gradients are bounded as 
\(
\mathbb{E}\big[\|\nabla F_i(w;\xi_i)\|^2\big] \le G_i^2
\)
and 
\(
\mathbb{E}\big[\|\nabla F_i(w;\xi_i)-\nabla F_i(w)\|^2\big] \le \chi_i^2
\), respectively.
\end{Assumption}
\noindent It is worth noting that Assumptions~\ref{ass:hetero_penalty} and \ref{ass:2} are widely used in stochastic optimization to derive explicit convergence-rate guarantees and sampling-dependent bounds \cite{joint_communications_FL, kaidi_FL, Lin_TWC_2024}.

\begin{Theorem} \label{the1}
Under i.i.d.\ $K$-client sampling with replacement using probabilities $q_i$, let $R_\epsilon$ denote the number of communication rounds required to reach a target accuracy $\epsilon$.
The following upper bound can be obtained:
\begin{equation}
\mathbb{E}[R_\epsilon]\le 
\sum_{i=1}^N\left(
\frac{\scoef}{\varepsilon}\frac{c_i}{q_i}
+ \frac{\sconst}{N\varepsilon}
\right).
\label{eq:round_bound}
\end{equation}
where
$\omega \triangleq \frac{E}{K}\frac{L_{\mathrm{sm}}}{m_{\mathrm{cv}}^2}$,
$c_i \triangleq p_i^2 G_i^2$, $V_{\mathrm{var}}\triangleq \sum_{i=1}^N p_i^2 \chi_i^2$, $V_{\mathrm{eng}}\triangleq \sum_{i=1}^N p_i G_i^2$. Then
\(
\sconst \triangleq 
\frac{\Ls}{\mc^2}\Big(\frac{V_{\mathrm{var}}}{E}+ F^\star-\sum_{i=1}^N p_i F_i^\star + E V_{\mathrm{eng}} \Big)
+\frac{\Ls}{\mc}\|\bw^0-\bw^\star\|,
\)
where $\bw^\star$ is any minimizer of $F$.
\end{Theorem}

\begin{proof}
    See Appendix~\ref{appendix:theorem1}.
\end{proof}

To connect the per-round straggler latency to wall-clock training time, let
$T_{\mathrm{round}}^{(r)}(q,x)$ denote the synchronous round time in round $r$ under fixed $(q,x)$.
Under fixed-design i.i.d.\ participation across rounds and time-invariant latencies $\{t_i(x)\}$,
the sequence $\{T_{\mathrm{round}}^{(r)}(q,x)\}_{r\ge 1}$ is i.i.d.\ with mean
$f(q,x)=\mathbb{E}[T_{\mathrm{round}}(q,x)]$. Moreover, since the event $\{R_\epsilon\ge r\}$ depends
only on the history up to round $r-1$ while $T_{\mathrm{round}}^{(r)}(q,x)$ depends only on the fresh
sampling in round $r$, Wald's identity gives
\[
\mathbb{E}\!\left[\sum_{r=1}^{R_\epsilon} T_{\mathrm{round}}^{(r)}(q,x)\right]
= f(q,x)\,\mathbb{E}[R_\epsilon].
\]
Combining this factorization with the bound in~\eqref{eq:round_bound} yields the expected time-to-accuracy upper bound
\[
\mathbb{E}[\text{time-to-}\epsilon]\le \frac{1}{\epsilon}\,
f(q,x)\underbrace{\sum_{i=1}^N\left(
\frac{\scoef c_i}{q_i}
+ \frac{\sconst}{N}
\right)}_{g(q)}.
\]
and the wall-clock training objective $J(q,x)\triangleq f(q,x)\,g(q)$.

\section{Optimal Participation for PA Location}
\label{sec:kkt}

Rapid model convergence in wireless FL necessitates a balance between two competing forces: minimizing the duration of each round and minimizing the total number of rounds required for training. This creates a fundamental tension between physical latency and statistical efficiency.
In this section we optimize the client sampling distribution for a fixed PA position $x$.
To simplify notation, we suppress the dependence on $x$ and write $t_i \triangleq t_i(x)$, $f(q)\triangleq f(q,x)$, and $J(q)\triangleq J(q,x)$ under the sorted order in \eqref{eq:sorted-t}.
We study the inner problem
\begin{equation}
    \textbf{(P1)}\qquad
    \min_{q\in\simplex}\ J(q)\triangleq f(q)g(q).
    \label{eq:P1}
\end{equation}
Problem (P1) is generally nonconvex due to the product coupling of a concave straggler functional 
$f(\cdot)$ and a convex heterogeneity penalty $g(\cdot)$.
\noindent A first question is whether an optimizer of \((\mathrm{P1})\) could lie on the simplex boundary, i.e., assign \(q_i=0\) to some clients, or whether the non-IID term forces every \(q_i\) to stay strictly positive. 

\begin{Lemma} \label{lem:interior}
Suppose $c_i>0$ for all $i$.
Then $J(q)=f(q)g(q)$ satisfies $\lim_{q_i\to 0^+}J(q)=+\infty$ for any $i$.
Consequently, problem \textbf{(P1)} admits at least one global minimizer $q^{\mathrm{opt}}\in\simplex$, and every global minimizer is bounded away from the boundary, in particular, it is an interior point.
\end{Lemma}
\begin{proof}
Refer to Appendix \ref{appendix:lemma2}.
\end{proof}
\noindent The next question is whether the exact order-statistics form can be rewritten into a representation that separates the contribution of each adjacent latency gap and makes the dependence on \(\{Q_i\}\) explicit.
To solve the inner minimization over $q$ for a fixed $x$, we rewrite $f(q)$ in a telescoping form over adjacent latency gaps, which makes its dependence on $q$ explicit.
Starting from \eqref{eq:f-order}, expand $\sum_{i=1}^N (Q_i^K-Q_{i-1}^K)t_i=t_N Q_N^K+\sum_{i=1}^{N-1}Q_i^K(t_i-t_{i+1})$:
\begin{equation}
    f(q)=t_N-\sum_{i=1}^{N-1}\Delta_i\,Q_i^K.
    \label{eq:f-telescope}
\end{equation}
where $\Delta_i\triangleq t_{i+1}-t_i\ge 0$ and $Q_N=1$.
In (\ref{eq:f-telescope}), \(Q_i^K\) is the probability that all \(K\) i.i.d.\ draws fall within the prefix \(\{1,\ldots,i\}\); thus each gap \(\Delta_i\) reduces the baseline \(t_N\) only when the sampled maximum does not cross that latency cliff (an event of probability \(Q_i^K\)). 
Equation~(\ref{eq:f-telescope}) shows that straggler latency is governed by a sum of latency ``cliffs'' $\Delta_i$, each discounted by the probability $Q_i^{K}$ that all $K$ samples fall within the faster prefix. This makes it explicit which gaps matter (those near the tail with $Q_i \approx 1$).
This raises a natural question: what is the marginal effect of increasing \(q_s\) on \(f(q)\), i.e., how sensitive is the tail latency to each coordinate of \(q\)? The following lemma can be obtained.

\begin{Lemma}\label{lem:f-concave}
Fix $x$ and consider the sorted latencies $t_1\le \cdots \le t_N$.
For $q\in\bar\Delta\triangleq\{q\in\mathbb{R}^N:\ q_i\ge 0,\ \sum_{i=1}^N q_i=1\}$, define
$Q_i\triangleq \sum_{j=1}^i q_j$ and $\Delta_i\triangleq t_{i+1}-t_i\ge 0$. Then the expected
straggler latency admits the gap form in $\eqref{eq:f-telescope}$.
Moreover: (i) The mapping $q\mapsto f(q)$ is concave on $\bar\Delta$ (and hence on $\Delta$). (ii) For any interior $q\in\Delta$, the directional derivative of $f$ along any feasible direction $v$ (where $\sum_{i=1}^N v_i=0$) is given by
\begin{equation}\label{eq:f-dirder}
\mathrm{D}f(q)[v]=-\sum_{s=1}^N D_s(q)\,v_s,
\qquad
D_s(q)\triangleq K\sum_{i=s}^{N-1}\Delta_i Q_i^{K-1}\ \ (\ge 0).
\end{equation}
This provides a convenient gradient representative on the simplex, such that one may take $\frac{\partial f}{\partial q_s}=-D_s(q)$.
(iii) For $K\ge 2$, $f$ is twice continuously differentiable on $\Delta$ with Hessian
\begin{equation}\label{eq:f-hessian}
\nabla^2 f(q)= -K(K-1)\sum_{i=1}^{N-1}\Delta_i Q_i^{K-2}\, a_i a_i^\top \preceq 0,
\qquad
a_i\triangleq ( \underbrace{1,\ldots,1}_{i},0,\ldots,0 )^\top .
\end{equation}
\end{Lemma}

\begin{proof}
Refer to Appendix \ref{appendix:lemma3}.
\end{proof}

Using this gradient structure, we derive the KKT system, whose recursion explicitly shows how the optimal $q$ trades off statistical weights $c_i$ against the marginal tail penalty induced by straggling. We now ask: compared with the statistics-only square-root rule \(q_i\propto\sqrt{c_i}\), how is the optimal sampling distribution distorted by straggler synchronization, and which latency gaps \(\Delta_i\) govern this distortion.
From \eqref{eq:f-telescope}, any feasible perturbation of $q$ satisfies $\sum_i\mathrm{d}q_i=0$ and hence $\mathrm{d}Q_N=0$, so the directional derivative of $f$ along the simplex equals the derivative of the reduced form \eqref{eq:f-telescope}:
$\partial f/\partial q_s=-D_s(q)$.
Also, $\frac{\partial g}{\partial q_s} = -\scoef\,\frac{c_s}{q_s^2}$.
The Lagrangian $\mathcal{L}(q,\mu)=f(q)g(q)+\mu(\sum_i q_i-1)$ yields the stationarity condition
$g\,\partial f/\partial q_s+f\,\partial g/\partial q_s+\mu=0$.
Substituting derivatives gives \eqref{eq:kkt-main} with $\lambda\triangleq \mu$.
For \eqref{eq:recursion}, subtract \eqref{eq:kkt-main} for $s=i$ and $s=i+1$ and use
$D_i(q)-D_{i+1}(q)=K\Delta_i Q_i^{K-1}$.

\begin{Theorem}
\label{thm:kkt}
Any interior KKT point $q^\star\in\simplex$ of (P1) satisfies, for some scalar $\lambda$,
\begin{equation}\label{eq:kkt-main}
f(q^\star)\,\scoef\,\frac{c_s}{(q_s^\star)^2}+g(q^\star)D_s(q^\star)=\lambda,
\qquad \forall s\in\mathcal{N}.
\end{equation}
where $D_s(q)\triangleq K\sum_{i=s}^{N-1}\Delta_i\,Q_i^{K-1}\ \ge 0$ denotes tail sensitivity.
Moreover, adjacent indices satisfy
\begin{equation}
    \frac{c_{i+1}}{(q_{i+1}^\star)^2}
    =
    \frac{c_i}{(q_i^\star)^2}
    +\frac{g(q^\star)}{f(q^\star)\scoef}\,K\,\Delta_i\,Q_i(q^\star)^{K-1},
    \qquad i=1,\dots,N-1.
    \label{eq:recursion}
\end{equation}
\end{Theorem}

\begin{Remark}
If one minimizes the statistical factor $g(q)$ alone, the KKT conditions equalize
$c_i/q_i^2$ across active clients, yielding the classical square-root rule $q_i\propto \sqrt{c_i}$.
Under the coupled wall-clock objective $J(q)=f(q)g(q)$, moving sampling mass from a faster client $i$ to a strictly
slower client $i+1$ must overcome an additional positive term:
\begin{equation} \label{rem:tailpremium}
\frac{c_{i+1}}{(q^\star_{i+1})^2}
=
\frac{c_i}{(q^\star_i)^2}
+
\underbrace{\frac{g(q^\star)}{f(q^\star)\omega}K\Delta_i Q_i(q^\star)^{K-1}}_{\text{tail-latency premium}},
\qquad i=1,\ldots,N-1. 
\end{equation}
Here $\Delta_i=t_{i+1}-t_i$ is the latency gap. The factor $Q_i(q^\star)^{K-1}$ admits a direct probabilistic
interpretation: it is the probability that the other $K-1$ i.i.d.\ draws lie in the prefix $\{1,\ldots,i\}$, so that the remaining draw becomes the decisive one that crosses the gap $\Delta_i$ and determines the round maximum; the multiplicative factor $K$ reflects the $K$ exchangeable opportunities for such a crossing.
\end{Remark}
This tail-latency premium is negligible unless $Q_i$ is close to one. In particular, the effective amplification scales as $KQ_i^{K-1}$, which remains order-one only when $Q_i=1-\Theta(1/K)$, foreshadowing the $\Theta(1/K)$ participation scale revealed by the two-class phase transition in Theorem~\ref{thm:threshold}.

Let $q^\star\in\Delta$ be any interior KKT point in Theorem~\ref{thm:kkt}.
Define $\psi_i\triangleq c_i/(q_i^\star)^2$. Then $\{\psi_i\}_{i=1}^N$ is nondecreasing given the recursion~\eqref{eq:recursion}:
\begin{equation} \label{eq:ri-recursion}
\psi_{i+1}-\psi_i=\frac{g(q^\star)}{f(q^\star)\scoef}\,K\,\Delta_i\,Q_i(q^\star)^{K-1}\ge 0,
\qquad i=1,\ldots,N-1,
\end{equation}
with strict inequality whenever $\Delta_i>0$.
Equivalently, $q_i^\star/\sqrt{c_i}$ is nonincreasing in $i$.
\begin{Remark}\label{rem:compare-sqrt}
Minimizing $g(q)$ alone yields the square-root rule $q_i\propto \sqrt{c_i}$.
Under the coupled objective, the recursion \eqref{eq:recursion} implies that
$\psi_i\triangleq c_i/(q_i^\star)^2$ is nondecreasing with the latency index $i$; equivalently,
$q_i^\star/\sqrt{c_i}$ is nonincreasing. Hence, slower clients are progressively reduced in relation to statistics-only allocation, with the distortion controlled by $K\Delta_i Q_i^{K-1}$.
\end{Remark}

\subsection{Class-wise Reduction}
When many clients share identical $t_i(x)$, we group them into latency classes, which enables a low-dimensional reformulation and yields a closed-form within-class allocation. The key question is whether, within a latency class where \(\Delta_i=0\), the KKT recursion simplifies to a tractable rule so that the N-dimensional vector $q$ can be reduced to class-level masses.

\begin{Lemma}
\label{lem:sqrt}
Consider a fixed $x$, the sorted times take $M$ distinct values $t^{(1)}<\cdots<t^{(M)}$, inducing a partition $\clients=\bigcup_{m=1}^M\mathcal{C}_m$ where $\mathcal{C}_m=\{i:\ t_i=t^{(m)}\}$. Then, for each class $m$, any KKT point $q^\star$ satisfies,
\begin{equation}
    q_i^\star=\delta_m^\star \frac{\sqrt{c_i}}{\sum_{j\in\mathcal{C}_m}\sqrt{c_j}},
    \qquad \forall i\in\mathcal{C}_m,
    \label{eq:sqrt-law}
\end{equation}
where $\delta_m^\star\triangleq\sum_{i\in\mathcal{C}_m}q_i^\star$ is the total probability mass assigned to class $m$.
\end{Lemma}
\begin{proof}
If $i$ and $i+1$ belong to the same class, then $\Delta_i=0$. The recursion \eqref{eq:recursion} reduces to $c_{i+1}/(q_{i+1}^\star)^2=c_i/(q_i^\star)^2$, implying $q_i^\star/\sqrt{c_i}$ is constant within the class. Normalization yields \eqref{eq:sqrt-law}.
\end{proof}
\noindent Under Lemma~\ref{lem:sqrt}, we can obtain the following Corollary.
\begin{Corollary}
\label{cor:class}
(P1) is equivalent to the $M$-dimensional problem
\begin{equation}
    \min_{\delta_1,\dots,\delta_M>0:\ \sum_m\delta_m=1}
    \Bigg(\sum_{m=1}^M\big(P_m^K-P_{m-1}^K\big)\,t^{(m)}\Bigg)
    \Bigg(\scoef\sum_{m=1}^M\frac{C_m^2}{\delta_m}+\sconst\Bigg),
    \label{eq:class-reduced}
\end{equation}
where $P_m\triangleq\sum_{\ell=1}^m\delta_\ell$ and $P_0\triangleq 0$, $
    C_m\triangleq\sum_{i\in\mathcal{C}_m}\sqrt{c_i}, \delta_m\triangleq\sum_{i\in\mathcal{C}_m}q_i, \sum_{m=1}^M\delta_m=1$.
\end{Corollary}
\begin{proof}
By Lemma~\ref{lem:sqrt}, $\sum_i c_i/q_i=\sum_m C_m^2/\delta_m$. The probability that the sampled maximum equals $t^{(m)}$ is $P_m^K-P_{m-1}^K$, giving the first factor.
\end{proof}

Having characterized $q^\star(x)$, we treat $J^\star(x)=J(q^\star(x),x)$ as piecewise smooth under a fixed latency ordering and derive its envelope derivative to obtain a first-order placement condition.
Before turning to the outer placement problem in Section~4, we further specialize the reduced form (\ref{eq:class-reduced}) to a two-class structure to expose an explicit participation scale and a phase transition in slow-class sampling.

\begin{Definition}
Let $\fast$ and $\slow$ be the fast/slow classes with times $t_f<t_s$ and gap $\Delta\triangleq t_s-t_f>0$. Define
\begin{equation}
    C_f\triangleq\sum_{i\in\fast}\sqrt{c_i},
    \qquad
    C_s\triangleq\sum_{i\in\slow}\sqrt{c_i},
    \qquad
    \delta\triangleq\sum_{i\in\slow} q_i\in(0,1).
\end{equation}
Throughout this subsection, we use $\delta$ without a subscript for simplicity of notation to denote the total probability mass assigned to the slow class, i.e. $\delta \triangleq \sum_{i \in \slow} q_i$.
\end{Definition}

Lemma~\ref{lem:sqrt} fixes the within-class form \eqref{eq:q-two-class} for any given $\delta$.
Substituting into $f$ yields $f(\delta)=t_f(1-\delta)^K+t_s\big(1-(1-\delta)^K\big)=t_s-\Delta(1-\delta)^K$.
Similarly, $\sum_i c_i/q_i=C_f^2/(1-\delta)+C_s^2/\delta$, giving \eqref{eq:Jdelta}.
Differentiation yields \eqref{eq:Jprime0}.

Under the two-class model, the optimal sampling distribution has the form
\begin{equation}
    q_i^\star=
    \begin{cases}
        \displaystyle \frac{1-\delta^\star}{C_f}\sqrt{c_i}, & i\in\fast,\\
        \displaystyle \frac{\delta^\star}{C_s}\sqrt{c_i}, & i\in\slow,
    \end{cases}
    \label{eq:q-two-class}
\end{equation}
where $\delta^\star\in(0,1)$ minimizes the scalar objective
\begin{equation}
    J(\delta)=\underbrace{\big[t_s-\Delta(1-\delta)^K\big]}_{f(\delta)}
    \underbrace{\Bigg[\scoef\Big(\frac{C_f^2}{1-\delta}+\frac{C_s^2}{\delta}\Big)+\sconst\Bigg]}_{g(\delta)}.
    \label{eq:Jdelta}
\end{equation}
Any interior stationary point satisfies
\begin{equation}
    \Delta K(1-\delta)^{K-1} g(\delta)
    +f(\delta)\,\scoef\Big(\frac{C_f^2}{(1-\delta)^2}-\frac{C_s^2}{\delta^2}\Big)=0.
    \label{eq:Jprime0}
\end{equation}

\begin{Proposition}
\label{prop:two-kkt}
Let $q^\star$ be any interior KKT point of the two-class problem. Then $\delta^\star$ satisfies the scalar condition
\begin{equation}
    \frac{C_s^2}{(\delta^\star)^2}-\frac{C_f^2}{(1-\delta^\star)^2}
    =\frac{g(\delta^\star)}{\scoef f(\delta^\star)}\,K\Delta\,(1-\delta^\star)^{K-1},
    \label{eq:two-kkt}
\end{equation}
with $f(\delta)$ and $g(\delta)$ defined in \eqref{eq:Jdelta}. In particular, the right-hand side is nonnegative, hence any KKT point satisfies
$C_s/(\delta^\star)\ge C_f/(1-\delta^\star)$.
\end{Proposition}
\begin{proof}
Refer to Appendix \ref{appendix:proposition1}.
\end{proof}

\noindent Finally, we ask how the optimal slow-class mass behaves as the per-round sample size \(K\) grows: does straggler synchronization force \(\delta^\star\) to shrink to the critical scale \(\Theta(1/K)\), and under what explicit condition does such a collapse occur for global minimizers?

We obtain the following theorem, which is stated for a sequence of problems indexed by $K$.
Since $C_s=\sum_{i\in\slow}\sqrt{c_i}$ aggregates the slow-class statistical weight, the condition
$C_s^2=O(1/K)$ should be read as a small-slow-class regime in which the slow class contributes
vanishing aggregated statistical benefit as $K$ grows, e.g., the slow clients carry a diminishing share of the target weights $p_i$, so that the $K$-fold tail-latency amplification can dominate and force $\delta^\star$ to the critical $1/K$ scale.

\begin{Theorem}
\label{thm:threshold}
Consider a sequence of two-class problems \eqref{eq:Jdelta} indexed by the per-round sample size $K\to\infty$.
Fix a constant $\rho>0$ and define $\delta_\rho\triangleq\rho/K$ and $P_{\rho,K}\triangleq(1-\rho/K)^K$.
Assume the slow-class aggregate satisfies the \,scaling\,
\begin{equation}
    C_s^2=O(K^{-1}),
\end{equation}
and, for some margin $\xi\in(0,1)$ and all sufficiently large $K$,
\begin{equation}
    C_s^2 \le (1-\xi)\,\frac{\rho}{K}\,
\frac{\Delta P_{\rho,K}}{t_s-\Delta P_{\rho,K}}\,
\frac{\scoef C_f^2+\sconst}{\scoef}.
    \label{eq:Cs-th}
\end{equation}
Then any sequence of global minimizers $\{\delta_K^\star\}$ of \eqref{eq:Jdelta} satisfies
\(
    \sup_{K}\,K\delta_K^\star <\infty,
\)
i.e., the slow-class total probability mass collapses from $\Theta(1)$ to $\delta_K^\star=O(1/K)$.
Moreover, the threshold \eqref{eq:Cs-th} is explicit in $(\{c_i\},\Delta,K,\scoef,\sconst)$ through $(C_f,C_s)$.
\end{Theorem}
\begin{proof}
Refer to Appendix \ref{appendix:theorem3}.
\end{proof}

\noindent Theorem~\ref{thm:threshold} formalizes a {participation collapse phenomenon driven by synchronous order statistics. For a fixed finite network, the theorem provides a sufficient condition for collapse at moderate $K$.
When $K$ grows, any constant slow-class mass $\delta=\Theta(1)$ makes the event ``at least one slow client is selected'' almost certain, so the expected round latency saturates at the slow-class time $t_s$.
To keep the expected straggler latency strictly below $t_s$, the slow-class mass must shrink to the critical scale $\delta=\Theta(1/K)$, for which $(1-\delta)^K$ remains order-one.

The condition \eqref{eq:Cs-th} is a sufficient threshold on the slow-class aggregate heterogeneity $C_s^2$ relative to the fast class and system parameters. 
Intuitively, if the slow class contributes relatively little to the non-IID penalty, i.e., small $C_s$, then its statistical benefit cannot compensate the $K$-amplified tail latency it induces, and the wall-clock optimal policy samples it only at rate $O(1/K)$.

\begin{Remark}
The phase transition is driven by order statistics.
Under two classes, the slow class becomes the straggler with probability $1-(1-\delta)^K$.
If $\delta=\Theta(1)$, then $(1-\delta)^K\to 0$ as $K\to\infty$ and $f(\delta)=t_s-\Delta(1-\delta)^K\to t_s$, i.e., the slow class is an almost-sure straggler and the wall-clock latency saturates at the slow-class round time.
To keep $f(\delta)$ strictly below $t_s$ as $K$ grows, the slow-class mass must shrink as $\delta=\rho/K$, for which $(1-\delta)^K\to e^{-\rho}$ and the slow class stops being an almost-sure maximum.
Theorem~\ref{thm:threshold} formalizes this intuition by providing an explicit sufficient condition under which any global minimizer satisfies $\delta^\star=O(1/K)$.
\end{Remark}

\begin{Remark}
The wall-clock optimal design may drive a slow class to $\delta^\star=\Theta(1/K)$, which is often indistinguishable from excluding the class in practice.
If slow clients correspond to long-tail data patterns needed for robust generalization, one can enforce coverage by adding a mild smoothing/coverage term (e.g., an entropy or KL regularizer) or by imposing lower-bound constraints such as $q_i\ge \eta p_i$ for some $\eta\in(0,1)$.
Such modifications preserve the same order-statistics latency structure in $f(\cdot)$ and alter the KKT system by adding an additional marginal term that counteracts collapse, yielding an explicit knob to trade off wall-clock speed against distributional coverage.
\end{Remark}

\section{Pinching-Antenna Placement: Exact Characterization and Global Search}
\label{sec:placement}

Section \ref{sec:kkt} characterized the optimal sampling distribution $q^*(x)$ for any fixed antenna position $x$. 
We now address the outer optimization: finding the globally optimal placement $x^*$ that minimizes the resulting wall-clock objective $J^*(x) \triangleq \min_{q \in \simplex} J(q,x)$.

The primary challenge arises from the non-smoothness of $J^*(x)$, due to discrete changes in the latency ordering \eqref{eq:sorted-t} as $x$ varies. 
Our strategy is to partition $[0,L]$ into finitely many regions (Lemma~\ref{lem:finite}) within which the ordering is fixed, 
apply an envelope-theorem derivative (Proposition~\ref{prop:envelope}), and search over stationary points and region boundaries (Algorithm~\ref{alg:global}).

\subsection{Envelope Derivative}
The PA position affects only $f(q,x)$ through the times $t_i(x)$.
Ordering changes (i.e., $t_i(x)=t_j(x)$) create nondifferentiable points because the sorted order \eqref{eq:sorted-t} changes.
Between such breakpoints, the ordering is fixed. For the placement derivative we need $t'_i(x)$. From (\ref{eq:ti}), we have 
\begin{equation}
    t_i'(x)=
    \frac{S_i}{B\ln 2}
    \frac{\gamma_i(x)}{(1+\gamma_i(x))\,\big[\log_2(1+\gamma_i(x))\big]^2}
    \Bigg(\frac{2(x-u_i)}{(x-u_i)^2+r_i^2}\Bigg).
    \label{eq:ti-deriv}
\end{equation}

\begin{Proposition}
\label{prop:envelope}
Fix an interval $\mathcal{I}\subset[0,L]$ on which the ordering \eqref{eq:sorted-t} does not change.
Let $q^\star(x)\in\argmin_{q\in\simplex}J(q,x)$.
Assume $q^\star(x)$ is chosen so that it varies continuously on $\mathcal{I}$ (e.g., the inner minimizer is unique).
Then for all $x\in\mathcal{I}$ where $t_i(x)$ are differentiable,
\begin{equation}
    \frac{\mathrm{d}}{\mathrm{d}x}J^\star(x)
    =g\big(q^\star(x)\big)\sum_{i=1}^N\big(Q_i(q^\star(x))^K-Q_{i-1}(q^\star(x))^K\big)\,t_i'(x).
    \label{eq:envelope}
\end{equation}
Consequently, any stationary point $x^\star\in\mathrm{int}(\mathcal{I})$ satisfies
\begin{equation}
    \sum_{i=1}^N\big(Q_i(q^\star(x^\star))^K-Q_{i-1}(q^\star(x^\star))^K\big)\,t_i'(x^\star)=0,
    \label{eq:xFOC}
\end{equation}
with $t_i'(x)$ given by \eqref{eq:ti-deriv}.
\end{Proposition}
\begin{proof}
On $\mathcal{I}$, the ordering \eqref{eq:sorted-t} is fixed, so the straggler expectation admits the smooth representation
$f(q,x)=\sum_{i=1}^N \pi_i(q)\,t_i(x)$ with weights $\pi_i(q)\triangleq Q_i^K-Q_{i-1}^K$ that depend only on $q$.
Therefore,
$\partial J/\partial x=g(q)\sum_i\pi_i(q)\,t_i'(x)$.

Although the inner minimization over $q$ is generally nonconvex due to the product coupling, $J(q,x)$ is continuously differentiable in $x$ on $\mathcal{I}$.
Under the standing assumption that a continuous selection $q^\star(x)\in\argmin_{q\in\simplex}J(q,x)$ exists on $\mathcal{I}$ (e.g., when the inner minimizer is unique), a standard envelope-theorem argument yields
$\frac{\mathrm{d}}{\mathrm{d}x}J^\star(x)=\left.\frac{\partial J(q,x)}{\partial x}\right|_{q=q^\star(x)}$,
which gives \eqref{eq:envelope}.
Setting \eqref{eq:envelope} to zero yields the stationary condition \eqref{eq:xFOC}.
\end{proof}


\begin{Remark}
The weight $\pi_i\big(q^\star(x)\big)=Q_i(q^\star(x))^K-Q_{i-1}(q^\star(x))^K$ equals the probability that client $i$ is the straggler (the round maximum) under sampling distribution $q^\star(x)$.
Thus \eqref{eq:envelope} states that the outer gradient is a weighted sum of physical gradients $t_i'(x)$, dominated by clients that are most likely to determine the synchronous round time.
\end{Remark}

We justify this piecewise analysis by showing that order changes occur only at finitely many pairwise crossings of $t_i(x)$, yielding a finite breakpoint partition of $[0,L]$.
Without loss of generality, we assume that for any pair $i\ne j$, the function $t_i(x)-t_j(x)$ is not identically zero on $[0,L]$.
\begin{Lemma}
\label{lem:finite}
Under the considered PASS model, for each pair $(i,j)$, the equation $t_i(x)=t_j(x)$ has finitely many solutions in $[0,L]$. Hence, the set of ordering breakpoints is finite.
\end{Lemma}
\begin{proof}
The functions $t_i(x)$ are real-analytic in $x$ on $[0,L]$ because they are compositions of analytic functions (polynomials, exponential, and $\log(1+\cdot)$).
Thus $t_i(x)-t_j(x)$ is analytic; if it is not identically zero, it has only finitely many zeros on the compact interval $[0,L]$.
\end{proof}

\begin{algorithm}[t]
\caption{Tail-Latency PASS Placement and Participation}
\label{alg:global}
\begin{algorithmic}[1]
\REQUIRE $\{u_i,r_i,S_i,t^{\mathrm{comp}}_i,c_i\}_{i=1}^N$, bandwidth $B$, $\scoef,\sconst$, $K$, interval $[0,L]$, tolerance $\eta$.
\ENSURE $x^*$ and sampling distribution $q^\star$.
\STATE Compute $t_i(x)$ by \eqref{eq:ti} and $t'_i(x)$ by \eqref{eq:ti-deriv}.
\STATE Initialize breakpoint set $\mathcal{B}\leftarrow\{0,L\}$.
\FOR{each pair $(i,j)$ with $i<j$}
    \STATE Numerically solve $t_i(x)=t_j(x)$ on $[0,L]$; add all roots to $\mathcal{B}$ (within tolerance $\eta$).
\ENDFOR
\STATE Sort $\mathcal{B}$ and form intervals $\mathcal{I}_m=[\mathcal{B}_m,\mathcal{B}_{m+1}]$.
\STATE Initialize best value $J_{\min}\leftarrow\infty$ and best solution $(x^\star,q^\star)\leftarrow(\mathrm{null},\mathrm{null})$.
\FOR{each interval $\mathcal{I}_m$}
    \STATE Pick $x_0\in\mathrm{int}(\mathcal{I}_m)$ and determine the fixed ordering of $\{t_i(x_0)\}$.
    \STATE \textbf{Inner solve:} for each queried $x$, compute minimizer $\hat q(x)$ of $\min_{q\in\Delta} J(q,x)$:
    \STATE \hspace{0.9cm} (a) if multiple clients have identical latencies, form latency classes and solve the reduced problem \eqref{eq:class-reduced};
    \STATE \hspace{0.9cm} (b) otherwise, solve \textbf{(P1)} using a multi-start strategy (e.g., projected gradient / SQP initialized from the statistical rule $q_i\propto\sqrt{c_i}$ and from perturbed variants), and keep the best solution found.
    \STATE \textbf{Outer root:} define $\phi(x)\triangleq\sum_i (Q_i(\hat q(x))^K-Q_{i-1}(\hat q(x))^K)\,t_i'(x)$.
    \STATE Find all roots of $\phi(x)=0$ in $\mathcal{I}_m$ (e.g., bracketing + bisection) and collect candidates $\mathcal{X}_m$.
    \STATE Evaluate $J(\hat q(x),x)$ for all $x\in\mathcal{X}_m\cup\{\mathcal{B}_m,\mathcal{B}_{m+1}\}$ and update $(x^\star,q^\star)$ if improved.

\ENDFOR
\RETURN $x^\star,q^\star$.
\end{algorithmic}
\end{algorithm}

\subsection{Global Placement Algorithm}
In this breakpoint partition, we evaluate stationary points and interval boundaries to
assemble a breakpoint-and-root candidate-enumeration algorithm for PA placement.
By Lemma~\ref{lem:finite}, $[0,L]$ can be partitioned into finitely many fixed-ordering intervals.
In each interval, Proposition~\ref{prop:envelope} reduces placement to 1D root finding.

Algorithm~1 is globally optimal for the outer envelope problem
$\min_{x\in[0,L]} J^\star(x)$ with $J^\star(x)=\min_{q\in\Delta}J(q,x)$, provided that the inner
participation problem is solved to global optimality for each queried $x$, and that all ordering
breakpoints and stationary candidates are enumerated within numerical tolerance.

\section{Simulation}\label{sec:simulation}
In this section, the simulation results are presented to validate the theoretical findings and evaluate the impact of the PASS on FL.
Unless otherwise stated, the system parameters are: bandwidth $B=10$\,MHz, transmit power $P=23$\,dBm, noise power spectral density $N_0=-174$\,dBm/Hz, and waveguide length $L=10$\,m \cite{Lin2025PASSFL,DingPoor2025Placement}.
We consider $N\in\{10,20,30\}$ clients with computation times drawn i.i.d.\ from $t_i^{\mathrm{comp}}\sim\mathrm{Unif}[0.05,0.15]$\,s.
To model non-IID data, we synthesize per-client label distributions using a Dirichlet model with concentration $\alpha_{\mathrm{data}}=0.3$; the heterogeneity score for each client is defined as the $\ell_2$ deviation from the global label distribution.
We implement synchronous FedAvg for $R=100$ rounds with per-round participation.
Each selected client performs $E=2$ local epochs of SGD with batch size 32, learning rate 0.01 \cite{Lin_TWC_2024, kaidi_FL}. The server aggregates client models by weighting proportional to local data size.
Test accuracy is evaluated every 5 rounds. For “$1\times$ data” and “$2\times$ data”, each client
holds 100 and 200 training samples, respectively.
We use a lightweight CNN: for MNIST, two $3\times3$ conv layers (32 and 64 channels) followed by
a 128-unit FC layer; for CIFAR-10, two $3\times3$ conv layers (64 and 128 channels) with $2\times2$
max-pooling after each conv, followed by a 256-unit FC layer \cite{Lin_twc_2025}.
\begin{figure}[t]
    \centering
    \includegraphics[width=1\linewidth]{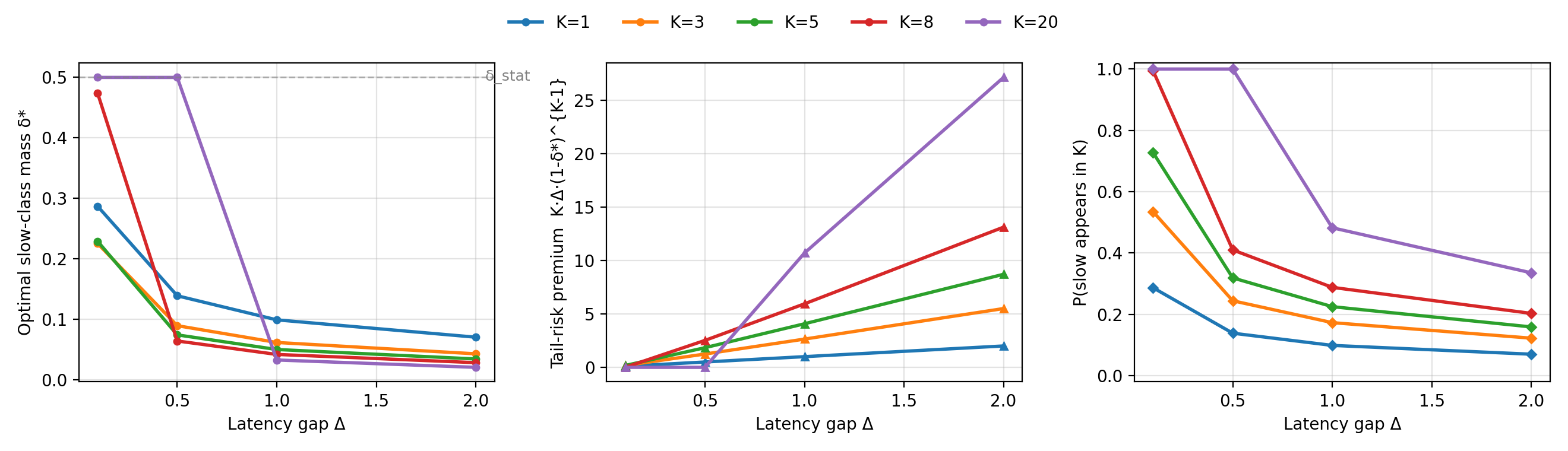}
    \caption{Tail-latency premium induced by synchronous straggling.
    (a) Optimal slow-class mass $\delta^\star$ versus latency gap $\Delta$.
    (b) Tail-latency premium $K\Delta(1-\delta^\star)^{K-1}$ at the optimizer.
    (c) Probability of sampling at least one slow client, $1-(1-\delta^\star)^K$.
    Larger $K$ amplifies the straggler penalty, driving $\delta^\star$ toward the $\Theta(1/K)$ scale.}
    \label{fig:exp1_tail_risk_premium}
\end{figure}

\begin{figure}[t]
    \centering
    \includegraphics[width=1\linewidth]{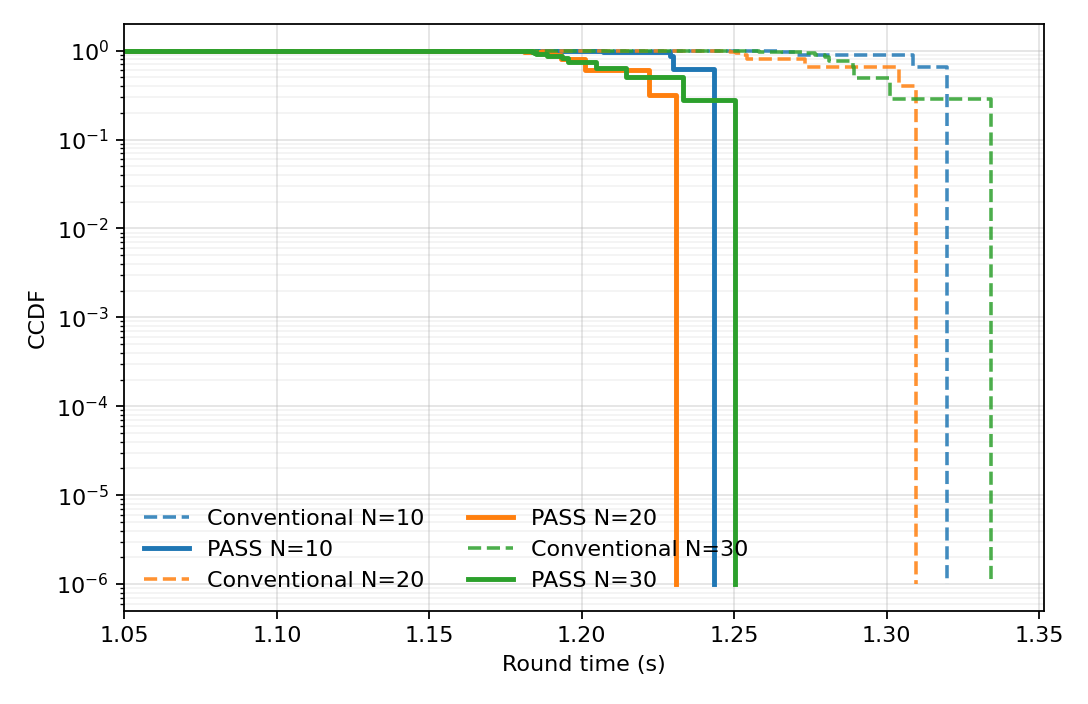}
    \caption{Complementary CDF of synchronous round time $T_{\mathrm{round}}$ under conventional fixed placement and PASS-optimized placement, for $N\in \{10, 20, 30\}$ clients. PASS shifts the tail left, reducing extreme straggler rounds that dominate wall-clock training.}
    \label{fig:ccdf}
\end{figure}

Fig.~\ref{fig:exp1_tail_risk_premium} illustrates how the optimal slow-class mass $\delta^\star$ responds to the latency gap $\Delta$.
Fig.~\ref{fig:exp1_tail_risk_premium}(a) shows that $\delta^\star$ decreases monotonically with $\Delta$ across all $K$: as the slow--fast gap widens, the wall-clock penalty for sampling slow clients increases, and the optimizer suppresses $\delta^\star$.
The decrease is sharper for larger $K$ because the probability of drawing at least one slow client grows with $K$, making straggling nearly certain unless $\delta^\star$ is reduced.
Fig.~\ref{fig:exp1_tail_risk_premium}(b) plots the tail-latency premium $K\Delta(1-\delta^\star)^{K-1}$---the marginal latency cost per unit increase in slow-class mass.
As $\Delta$ grows and $\delta^\star$ shrinks, $(1-\delta^\star)^{K-1}\approx 1$, yielding an effective amplification close to $K\Delta$.
This is the regime where order-statistics straggling most strongly distorts the allocation: small changes in $\delta$ have large wall-clock consequences.
Fig.~\ref{fig:exp1_tail_risk_premium}(c) shows the probability $1-(1-\delta^\star)^K$ that at least one slow client is sampled.
The optimizer implicitly targets a non-saturating straggler probability by shrinking $\delta^\star$ as $\Delta$ increases.
The scaling $\delta=\rho/K$ yields $(1-\delta)^K\to e^{-\rho}$, explaining why the $\Theta(1/K)$ participation scale emerges naturally.

\begin{figure}[t]
    \centering
    \includegraphics[width=1\linewidth]{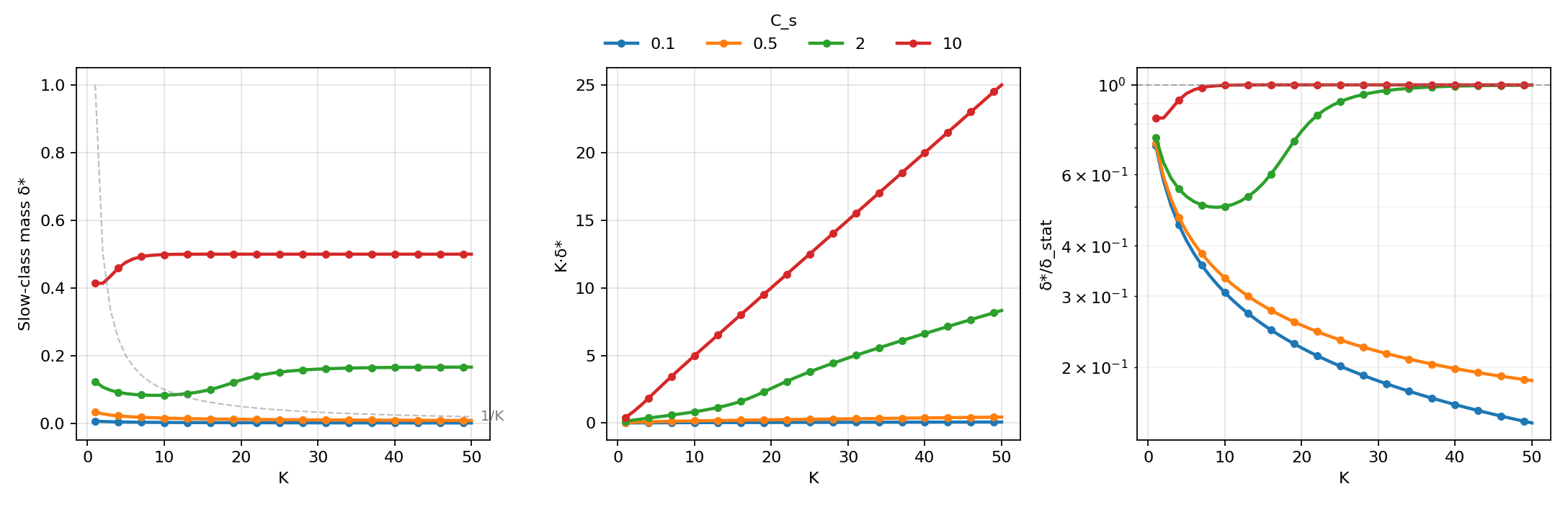}
    \caption{Phase transition of slow-class participation.
    (a) Optimal $\delta^\star$ versus $K$ for different $C_s$ values.
    (b) Normalized product $K\delta^\star$: bounded behavior indicates $\delta^\star=\Theta(1/K)$; linear growth indicates $\delta^\star=\Theta(1)$.
    (c) Ratio $\delta^\star/\delta^\star_{\mathrm{stat}}$, quantifying deviation from the statistics-only optimum.}
    \label{fig:exp2_phase_transition}
\end{figure}

Fig.~\ref{fig:ccdf} compares the complementary CDF $\mathbb{P}(T_{\mathrm{round}}>t)$ for conventional fixed placement (dashed) and PASS-optimized placement (solid).
Across exceedance probabilities down to $\sim 10^{-6}$, PASS exhibits a systematic left shift, reducing both mean round time and the high quantiles that dominate wall-clock performance under synchronization.
The persistent separation deep into the tail indicates that PASS improves the tail structure, reducing extreme straggler rounds rather than merely shifting the mean.
As $N$ increases, the conventional tail degrades more noticeably, while PASS remains stable, consistent with PASS mitigating worst-link outcomes that govern the maximum-order statistic.

Fig.~\ref{fig:exp2_phase_transition} validates the phase transition predicted by Theorem~\ref{thm:threshold}.
Panel~(a) shows that for smaller $C_s$, $\delta^\star$ decreases rapidly with $K$ and approaches the $1/K$ reference, while for larger $C_s$, $\delta^\star$ remains $O(1)$---the slow class's statistical value justifies frequent inclusion despite its latency cost.
Panel~(b) plots $K\delta^\star$: bounded (flat) behavior confirms $\delta^\star=\Theta(1/K)$, while linear growth indicates $\delta^\star=\Theta(1)$.
Panel~(c) shows $\delta^\star/\delta^\star_{\mathrm{stat}}$, the ratio to the statistics-only optimum.
For small $C_s$, this ratio decreases with $K$, indicating that tail latency increasingly dominates; for large $C_s$, it remains near unity, indicating statistical considerations prevail.

\begin{figure}[t]
    \centering
    \includegraphics[width=1\linewidth]{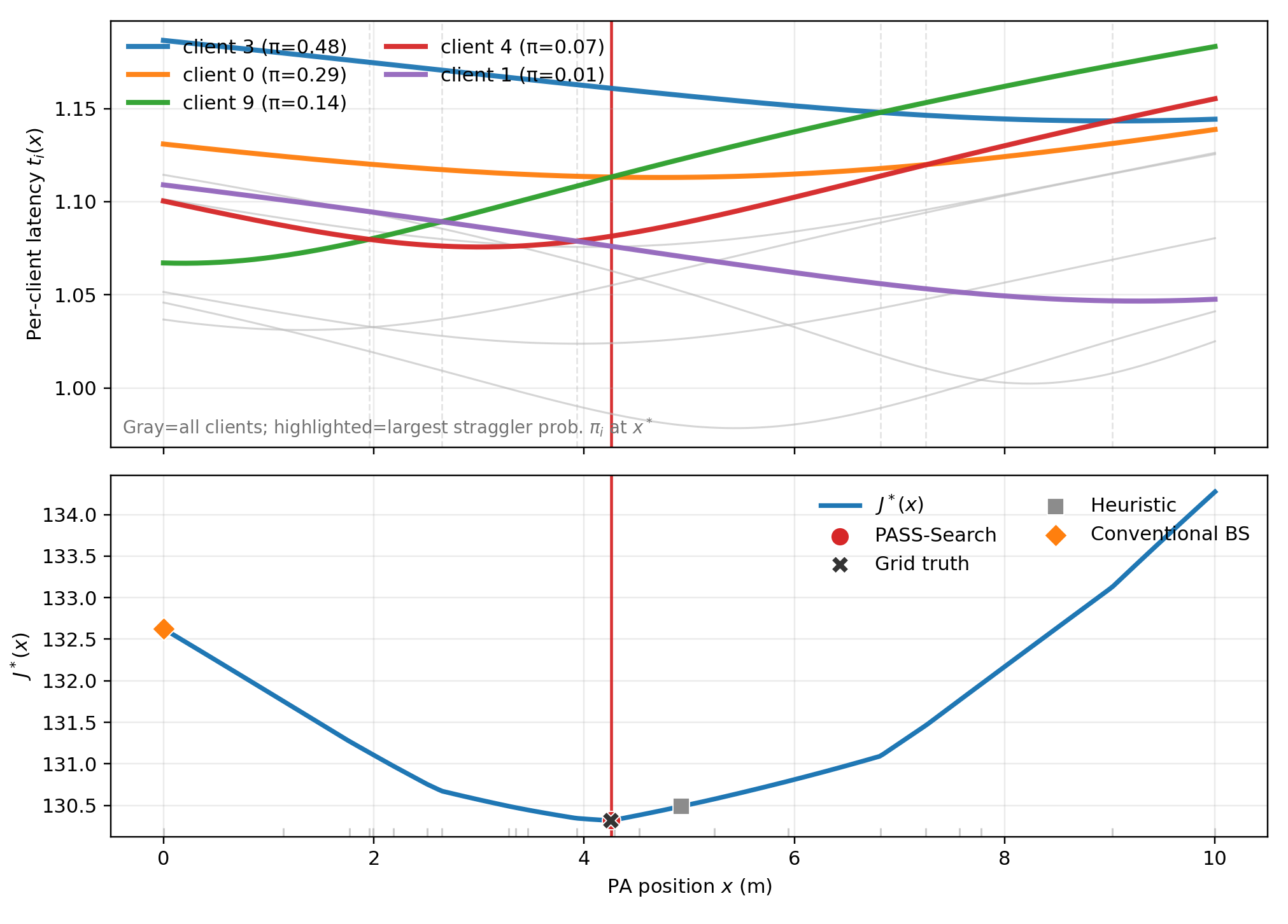}
    \caption{Piecewise-smooth envelope objective versus PA position.
    Top panel: per-client latency profiles $t_i(x)$; crossings induce ordering changes.
    Bottom panel: envelope objective $J^\star(x)$; breakpoints occur where $t_i(x)=t_j(x)$.
    At the optimum $x^\star$ (red line), straggling concentrates on different clients.}
    \label{fig:exp3_breakpoints_mechanism}
\end{figure}

Fig.~\ref{fig:exp3_breakpoints_mechanism} illustrates the piecewise structure of the placement problem.
The upper panel shows per-client latency profiles $t_i(x)$; curve crossings change the latency ordering and induce breakpoints.
The lower panel shows the envelope objective $J^\star(x)=\min_{q\in\simplex}J(q,x)$: within any interval where the ordering is fixed, $J^\star(x)$ varies smoothly, but its slope changes at breakpoints where $t_i(x)=t_j(x)$.

At the optimum $x^\star$ (red vertical line), straggling is concentrated on a few clients: $\pi_3\approx0.48$, $\pi_0\approx0.29$, $\pi_9\approx0.14$, $\pi_4\approx0.07$, and $\pi_1\approx0.01$.
This concentration yields an interpretable stationarity condition: optimality requires $\phi(x)=\sum_i\pi_i(x)\,t_i'(x)\approx 0$, i.e., the dominant straggler curves contribute opposing local trends that cancel near $x^\star$.

PASS-Search (red marker) matches the grid-based solution (black $\times$) while requiring far fewer evaluations.
In contrast, a geometry-only heuristic (gray square) and the conventional baseline at $x=0$ (orange diamond) incur higher $J^\star(x)$, demonstrating that performance gains require jointly accounting for order-statistics straggling and the breakpoint-induced objective geometry.

\begin{figure}[t]
    \centering
    \includegraphics[width=1\linewidth]{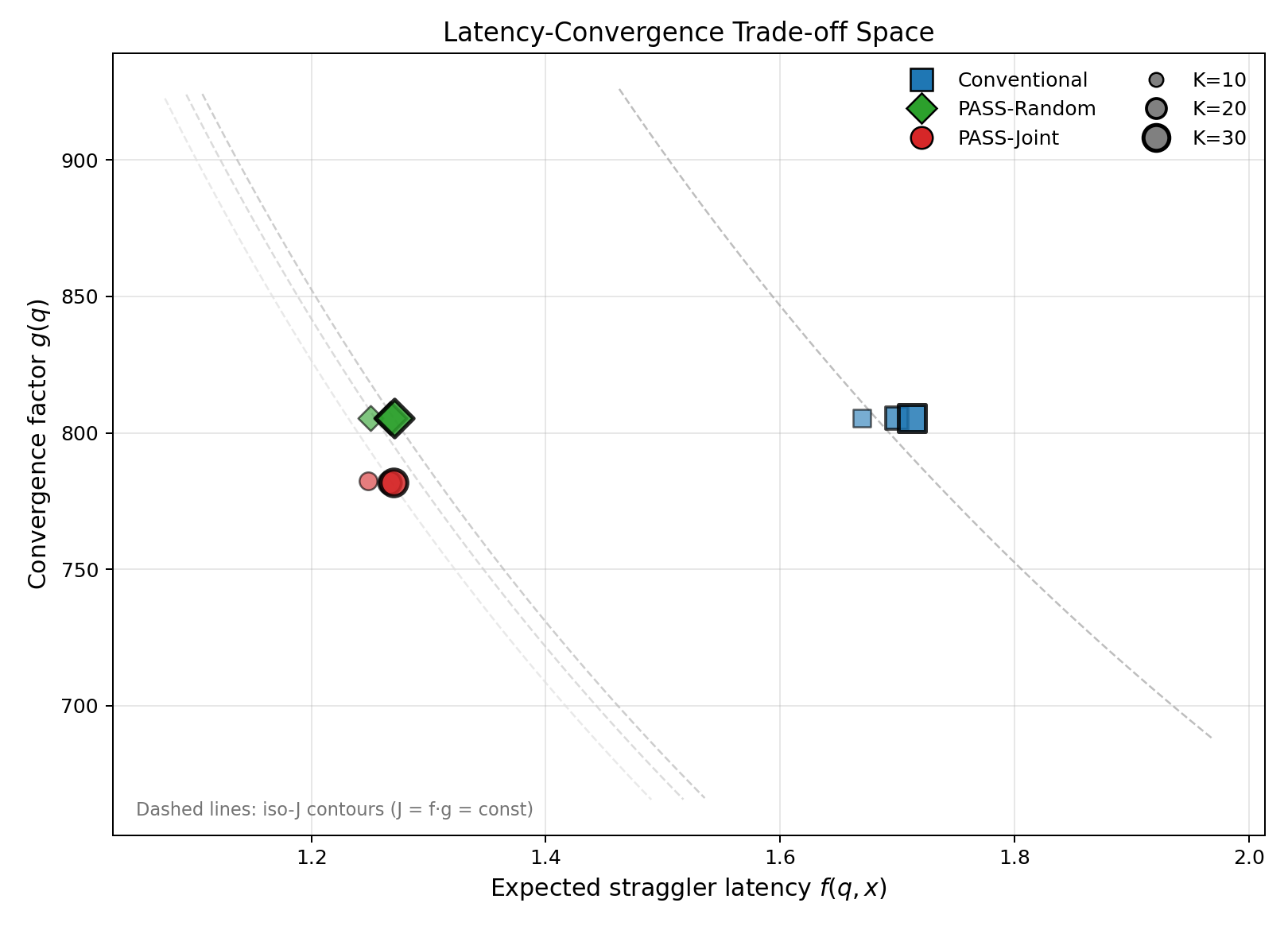}
    \caption{Latency--convergence trade-off space. Each marker represents a design point in the $(f,g)$ plane, where $f(q,x)$ is the expected straggler latency and $g(q)$ is the convergence factor. Marker shape indicates the method (Conventional, PASS-Random, PASS-Joint) and marker size indicates sample size $K\in\{10,20,30\}$.}
    \label{fig:exp12_fg_tradeoff}
\end{figure}


\begin{figure}[t]
    \centering
    \includegraphics[width=1\linewidth]{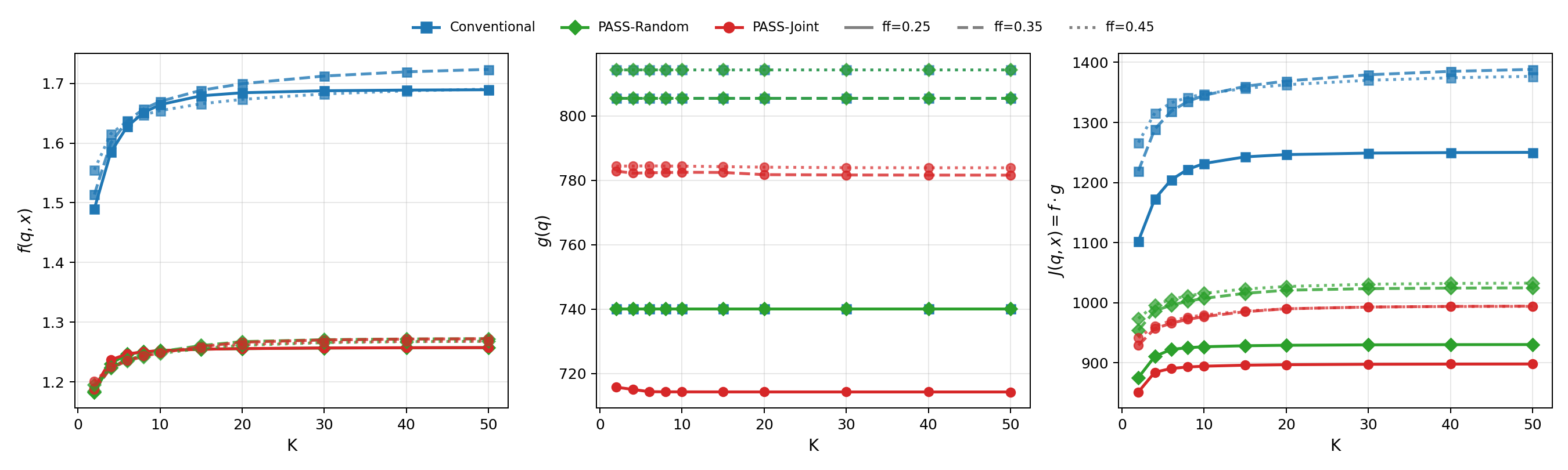}
    \caption{Decomposition of the wall-clock objective $J=f\cdot g$ versus sample size $K$, under different fast-class fractions $f_f\in\{0.25,0.35,0.45\}$. 
    Left: expected straggler latency $f(q,x)$. 
    Center: convergence factor $g(q)$. 
    Right: product $J(q,x)=f\cdot g$. 
    Conventional placement exhibits increasing $f$ with $K$ due to order-statistics amplification; PASS-Joint maintains nearly flat $f$ by optimizing PA position. The convergence factor $g$ is relatively stable across methods. Consequently, the product $J$ grows significantly for Conventional but remains low and stable for PASS-Joint.}
    \label{fig:exp12_fg_vs_k}
\end{figure}


\begin{figure}[t]
    \centering
    \subfloat[\textbf{MNIST.}\label{fig:mnist}]{
        \includegraphics[width=1\linewidth]{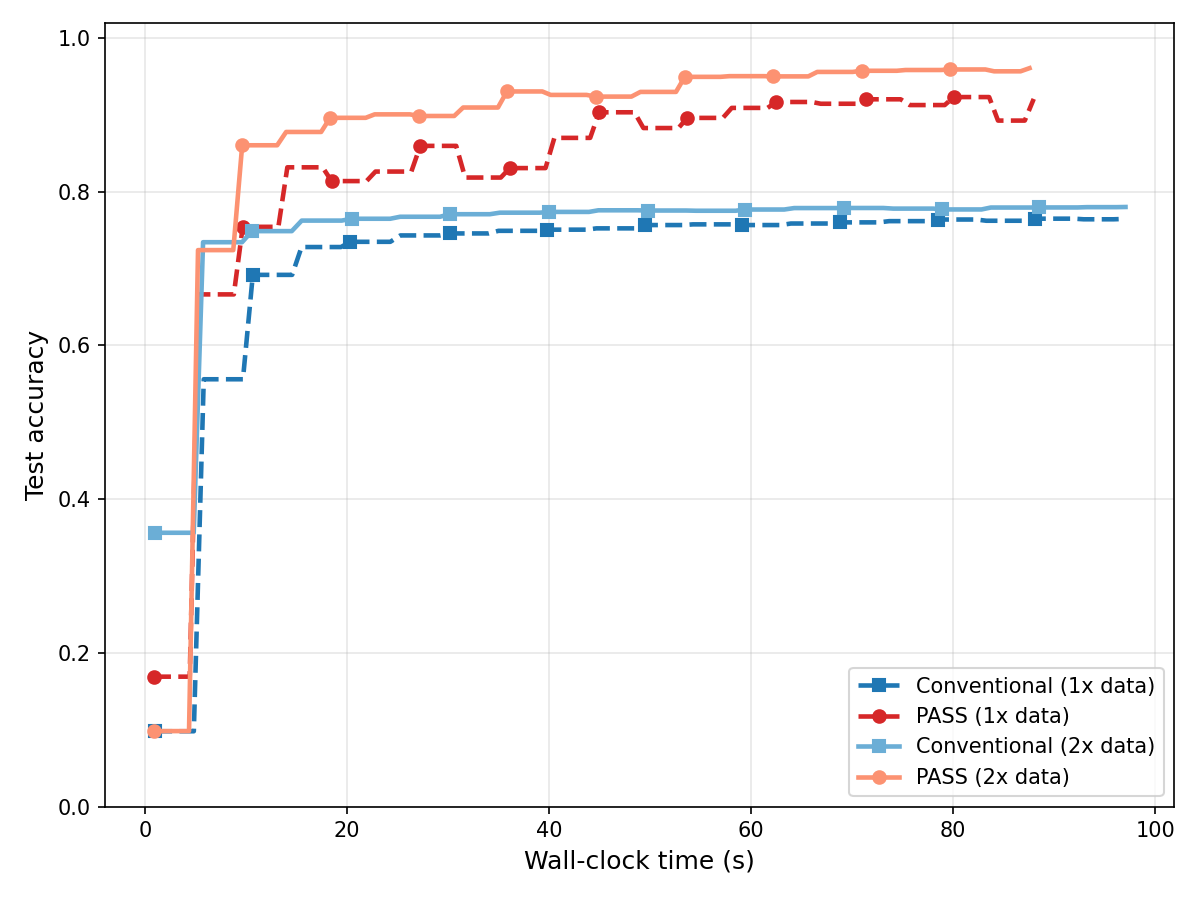}
    }
    \hfill
    \subfloat[\textbf{CIFAR-10.}\label{fig:cifar}]{
        \includegraphics[width=1\linewidth]{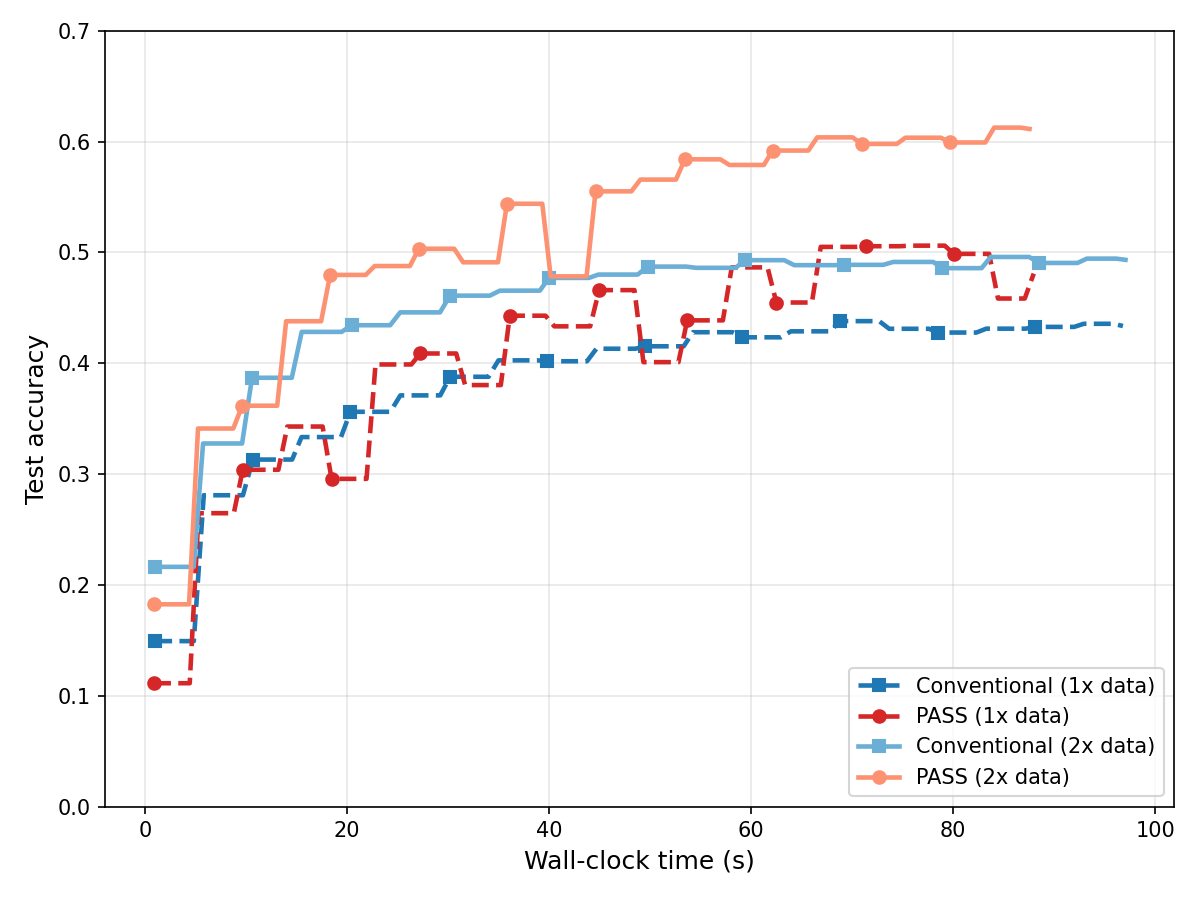}
    }
    \caption{Test accuracy versus wall-clock time under deadline-constrained synchronous FL. Solid lines denote PASS-optimized placement; dashed lines denote conventional fixed placement. ``$1\times$ data'' and ``$2\times$ data'' indicate the relative local dataset size per client. On both datasets, PASS achieves higher accuracy at any given wall-clock budget, with the gap widening under increased data (where more rounds are needed to converge). PASS (2$\times$ data) reaches approximately $95\%$ on MNIST and $61\%$ on CIFAR-10, compared to $77\%$ and $49\%$ for Conventional (2$\times$ data), respectively.}
    \label{fig:exp16_deadline}
\end{figure}



Fig.~\ref{fig:exp12_fg_tradeoff} visualizes the latency--convergence trade-off in the $(f,g)$ plane.
Each design point corresponds to a specific method and sample size $K$: Conventional (fixed placement with uniform sampling), PASS-Random (optimized PA placement with uniform sampling), and PASS-Joint (joint optimization of PA placement and participation).
The dashed iso-$J$ contours show level sets of the wall-clock objective $J=f\cdot g$; points closer to the origin lie on lower contours and achieve better time-to-accuracy.
Conventional designs cluster in the upper-right region ($f\approx 1.7$, $g\approx 805$), reflecting high straggler latency under fixed placement.
PASS-Random shifts points left ($f\approx 1.25$) by optimizing PA position, while PASS-Joint further reduces $g$ (to $\approx 780$) by jointly tuning participation.
Across all $K$, PASS-Joint consistently occupies the lowest iso-$J$ contours, demonstrating that joint optimization yields Pareto-dominant operating points in the latency--convergence space.


Fig.~\ref{fig:exp12_fg_vs_k} decomposes the wall-clock objective into its constituent factors as $K$ varies.
The left panel shows that the expected straggler latency $f(q,x)$ increases with $K$ under Conventional placement: as more clients are sampled, the probability of drawing a slow client (and thus incurring the worst-case latency) grows, consistent with the order-statistics analysis in Section~\ref{sec:kkt}.
In contrast, PASS-Joint maintains a nearly flat $f$ across $K$ by repositioning the PA to shorten the worst links.
The center panel shows that the convergence factor $g(q)$ is relatively insensitive to the method, since all designs use similar participation distributions that satisfy non-IID coverage requirements.
The right panel plots the product $J=f\cdot g$: Conventional exhibits a $20$--$30\%$ increase in $J$ as $K$ grows from $1$ to $50$, while PASS-Joint remains stable, confirming that PASS mitigates the $K$-amplified tail penalty identified theoretically.


Fig.~\ref{fig:exp16_deadline} evaluates end-to-end FL training under a wall-clock budget, comparing Conventional and PASS on MNIST and CIFAR-10.
Each curve tracks test accuracy as training progresses in real time, with solid lines indicating PASS-optimized placement and dashed lines indicating conventional fixed placement.
On MNIST (Fig.~\ref{fig:mnist}), PASS (2$\times$ data) reaches approximately $95\%$ accuracy within $20$\,s and maintains this level, whereas Conventional (2$\times$ data) plateaus around $77\%$.
The gap is even more pronounced on CIFAR-10 (Fig.~\ref{fig:cifar}), where PASS achieves around $61\%$ accuracy compared to around $49\%$ for Conventional under the same data and time budget.
These results confirm that the latency reduction provided by PASS translates directly into faster convergence and higher final accuracy under synchronous FL with deadline constraints: by reducing straggler rounds, PASS completes more communication rounds within the same wall-clock window, accumulating greater statistical progress.

\section{Conclusions} \label{sec:conclusion}

This paper developed a joint theory of client participation and pinching-antenna placement for PASS-enabled FL under non-IID data.
The KKT recursion reveals how latency gaps induce $K$-amplified tail penalties in the optimal sampling distribution, distorting the classical square-root rule.
Under latency-class structure, we established within-class square-root sampling and derived a two-class phase-transition threshold under which slow-class probability mass collapses to the $O(1/K)$ scale.
For antenna placement, we characterized the piecewise envelope structure and provided an exact breakpoint-and-root candidate-enumeration algorithm.
Simulations confirm that the joint optimization yields Pareto-dominant operating points in the latency--convergence trade-off and translates to substantial accuracy gains.
Future directions include multi-PA architectures with coordinated placement across segmented waveguides, and asynchronous FL where clients contribute under bounded staleness rather than strict synchronization.
\appendix

\section{Proof of Theorem \ref{the1}} \label{appendix:theorem1}
Under Assumption~\ref{ass:hetero_penalty},
the standard descent analysis \cite{joint_communications_FL, Lin2025PASSFL} for $K$-client sampling with $E$ local SGD steps yields the well-known $O(1/R)$ bound, where the \emph{only} $q$-dependent term enters through $\sum_i c_i/q_i$.
Specifically, letting $\bw^{r}$ be the global iterate and using the second-moment bound,
we obtain
\begin{equation}\label{eq:plug_descent}
\mathbb{E}\!\left[F(\bw^{R})\right]-F^\star
\;\le\;
\frac{1}{R}\left(\omega\sum_{i=1}^N \frac{c_i}{q_i}+\nu\right),
\end{equation}
Therefore, choosing $R=R_\varepsilon$ such that the RHS of \eqref{eq:plug_descent} is at most $\varepsilon$ gives
\begin{equation}\label{eq:round_bound_from_plug}
\mathbb{E}[R_\varepsilon]
\;\le\;
\frac{1}{\varepsilon}\Bigg(
\scoef\sum_{i=1}^N\frac{c_i}{q_i}
+\sconst
\Bigg),
\end{equation}
which matches \eqref{eq:round_bound}.

\section{Proof of Lemma \ref{lem:interior}} \label{appendix:lemma2}
For fixed $x$, the per-round latencies are positive constants with $0<t_1\le t_i\le t_N$.
Hence the straggler expectation satisfies $t_1\le f(q)\le t_N$ for all $q\in\simplex$.
Moreover, $g(q)=\scoef\sum_i c_i/q_i+\sconst\ge \scoef c_s/q_s$ for any index $s$.
Thus if $q_s\to 0^+$, then $g(q)\to+\infty$ and $J(q)\ge t_1 g(q)\to+\infty$.

\section{Proof of Lemma \ref{lem:f-concave}} \label{appendix:lemma3}
First, consider the function $\phi(u)=u^K$ on $u\ge 0$. For integer $K\ge 1$, $\phi''(u)=K(K-1)u^{K-2}\ge 0$,
so $\phi$ is convex on $[0,\infty)$. $Q_i(q)=\sum_{j=1}^i q_j$ is affine in $q$, hence
$Q_i(q)^K=\phi(Q_i(q))$ is convex in $q$. Since $\Delta_i\ge 0$, the term $-\Delta_i Q_i(q)^K$ is concave. Summing
over $i=1,\ldots,N-1$ and adding the constant $t_N$ establishes the concavity of $f$ on $\bar\Delta$.

Second, fix $q\in\Delta$ and a feasible direction $v$ with $\sum_{i=1}^N v_i=0$. Then
$\mathrm{D}Q_i(q)[v]=\sum_{j=1}^i v_j$. Applying the chain rule, we obtain
$\mathrm{D}f(q)[v] = -\sum_{i=1}^{N-1}\Delta_i\cdot K Q_i^{K-1}\cdot \sum_{j=1}^i v_j$.
Reordering this double sum yields
$\mathrm{D}f(q)[v] = -\sum_{s=1}^{N}\Big(K\sum_{i=s}^{N-1}\Delta_i Q_i^{K-1}\Big)v_s = -\sum_{s=1}^N D_s(q)\,v_s$,
which confirms~\eqref{eq:f-dirder}.

Finally, for $K\ge 2$, differentiating~\eqref{eq:f-dirder} once more yields the Hessian representation \eqref{eq:f-hessian}. Since each $a_i a_i^\top\succeq 0$ and each coefficient $\Delta_i Q_i^{K-2}\ge 0$, we have $\nabla^2 f(q)\preceq 0$, consistent with the concavity property derived in (i).

\section{Proof of Proposition \ref{prop:two-kkt}} \label{appendix:proposition1}
In the sorted order, all fast clients precede slow clients. There is exactly one nonzero gap at the class boundary: $\Delta_{n_f}=\Delta$, where $n_f\triangleq|\fast|$.
Thus, for any slow index $s>n_f$, $D_s=0$; for any fast index $s\le n_f$, $D_s=K\Delta Q_{n_f}^{K-1}=K\Delta(1-\delta)^{K-1}$.

By the within-class square-root form, for any fast client $i$, $c_i/(q_i^\star)^2=C_f^2/(1-\delta^\star)^2$, and for any slow client $j$, $c_j/(q_j^\star)^2=C_s^2/(\delta^\star)^2$.
Applying \eqref{eq:kkt-main} to any slow client gives $\lambda=f\scoef C_s^2/(\delta^\star)^2$.
Applying \eqref{eq:kkt-main} to any fast client and substituting $\lambda$ yields \eqref{eq:two-kkt}.

Let $\bar{\simplex}\!=\!\{q\in\R^N:\ q_i\ge 0,\ \sum_i q_i=1\}$.
Since $J(q)\to+\infty$ near the boundary of $\bar{\simplex}$, any minimizing sequence is eventually contained in a compact subset $\{q\in\bar{\simplex}:\ q_i\ge \varepsilon\ \forall i\}$ for some $\varepsilon>0$.
By compactness and continuity of $J$ on this subset, a global minimizer exists and must satisfy $q_i>0$ for all $i$, i.e., $q^{\mathrm{opt}}\in\simplex$.

\section{Proof of Theorem \ref{thm:threshold}} \label{appendix:theorem3}
The key observation is that the tail term $(1-\delta)^K$ becomes negligible as soon as $K\delta\to\infty$.
Assume, for contradiction, that there exists a subsequence $\{K_n\}$ such that the corresponding minimizers satisfy $K_n\delta_{K_n}^\star\to\infty$.
Then $(1-\delta_{K_n}^\star)^{K_n}\le \exp\{-K_n\delta_{K_n}^\star\}\to 0$ and hence
\begin{equation}
    f(\delta_{K_n}^\star)=t_s-\Delta(1-\delta_{K_n}^\star)^{K_n}\to t_s.
\end{equation}
On the other hand, for all $\delta\in(0,1)$ we have the uniform lower bound
\begin{equation}
    g(\delta)=\scoef\Big(\frac{C_f^2}{1-\delta}+\frac{C_s^2}{\delta}\Big)+\sconst
    \ge \scoef C_f^2+\sconst.
\end{equation}
Therefore,
\begin{equation}
    \liminf_{n\to\infty} J_{K_n}(\delta_{K_n}^\star)
    =\liminf_{n\to\infty} f(\delta_{K_n}^\star)g(\delta_{K_n}^\star)
    \ge t_s(\scoef C_f^2+\sconst).
    \label{eq:liminf-const}
\end{equation}
Now evaluate the objective at the tail test point $\delta_\rho=\rho/K$:
\begin{equation}
    J_K(\delta_\rho)=\big[t_s-\Delta P_{\rho,K}\big]\Bigg[\scoef\Big(\frac{C_f^2}{1-\rho/K}+\frac{C_s^2K}{\rho}\Big)+\sconst\Bigg].
    \label{eq:Jdrho}
\end{equation}
Using \eqref{eq:Cs-th}, we bound the slow-class term as
$\scoef(C_s^2K/\rho)\le (1-\xi)(\scoef C_f^2+\sconst)\,\Delta P_{\rho,K}/(t_s-\Delta P_{\rho,K})$.
Substituting this into \eqref{eq:Jdrho} yields
\begin{equation}
    J_K(\delta_\rho)
    \le \big[t_s-\Delta P_{\rho,K}\big]\Big[\scoef\tfrac{C_f^2}{1-\rho/K}+\sconst\Big]
    +(1-\xi)(\scoef C_f^2+\sconst)\,\Delta P_{\rho,K}.
\end{equation}
Since $\rho$ is fixed and $K\to\infty$, we have $\tfrac{1}{1-\rho/K}=1+O(1/K)$ and $P_{\rho,K}\to e^{-\rho}$. Consequently,
there exists $K_0$ such that for all $K\ge K_0$,
\begin{equation}
    J_K(\delta_\rho)
    \le t_s(\scoef C_f^2+\sconst)-\tfrac{\xi}{2}(\scoef C_f^2+\sconst)\,\Delta P_{\rho,K}.
    \label{eq:strict-improve}
\end{equation}
Combining \eqref{eq:strict-improve} with \eqref{eq:liminf-const} contradicts the optimality of $\delta_{K_n}^\star$ for all sufficiently large $n$.
Hence, $K\delta_K^\star$ cannot diverge; equivalently, $\sup_K K\delta_K^\star<\infty$, i.e., $\delta_K^\star=O(1/K)$.

~\bibliographystyle{IEEEtran}
\bibliography{references}
\end{document}